\documentclass[sigconf]{acmart}
\AtBeginDocument{%
  }

\setcopyright{acmlicensed}
\copyrightyear{2018}
\acmYear{2018}
\acmDOI{XXXXXXX.XXXXXXX}
\acmConference[Conference acronym 'XX]{Make sure to enter the correct
  conference title from your rights confirmation email}{June 03--05,
  2018}{Woodstock, NY}
\acmISBN{978-1-4503-XXXX-X/2018/06}

\usepackage{multirow} 
\usepackage{algorithm}
\usepackage{algpseudocode}

\begin{document}

\title{Dual-Rerank: Fusing Sequential Dependencies and Utility for Generative Reranking}

\author{Chao Zhang}
\authornote{Both authors contributed equally to this research.}
\affiliation{%
  \institution{Kuaishou Technology}
  \city{Hangzhou}
  \country{China}}
\email{zhangchao29@kuaishou.com}
\orcid{0009-0008-4098-3298}

\author{Shuai Lin}
\authornotemark[1]
\affiliation{%
  \institution{Kuaishou Technology}
  \city{Hangzhou}
  \country{China}}
\email{linshuai@kuaishou.com}


\author{Chenglei Dai}
\authornote{Corresponding Authors}
\affiliation{%
  \institution{Kuaishou Technology}
  \city{Hangzhou}
  \country{China}}
\email{daichenglei@kuaishou.com}

\author{Ye Qian}
\affiliation{%
  \institution{Kuaishou Technology}
  \city{Hangzhou}
  \country{China}}
\email{qianye@kuaishou.com}

\author{Mingyang Fan}
\affiliation{%
  \institution{Kuaishou Technology}
  \city{Beijing}
  \country{China}}
\email{fanmingyang03@kuaishou.com}





\author{Yi Zhang}
\affiliation{%
  \institution{Kuaishou Technology}
  \city{Beijing}
  \country{China}}
\email{zhangyi49@kuaishou.com}

\author{Yi Wang}
\affiliation{%
  \institution{Kuaishou Technology}
  \city{Beijing}
  \country{China}}
\email{wangyi05@kuaishou.com}

\author{Jingwei Zhuo}
\authornotemark[2]
\affiliation{%
  \institution{Unaffiliated}
  \city{Beijing}
  \country{China}}
\email{zhuojw10@gmail.com}

\renewcommand{\shortauthors}{Chao Zhang et al.}

\begin{abstract} 
Kuaishou serves over 400 million daily active users, processing hundreds of millions of search queries daily against a repository of tens of billions of short videos. 
As the final decision layer, the reranking stage determines user experience by optimizing whole-page utility.
While traditional score-and-sort methods fail to capture combinatorial dependencies, Generative Reranking offers a superior paradigm by directly modeling the permutation probability.
However, deploying Generative Reranking in such a high-stakes environment faces a fundamental dual dilemma: \textbf{1)} the \textit{structural trade-off} where Autoregressive (AR) models offer superior Sequential modeling but suffer from prohibitive latency, versus Non-Autoregressive (NAR) models that enable efficiency but lack dependency capturing; \textbf{2)} the \textit{optimization gap} where Supervised Learning faces challenges in directly optimizing whole-page utility, while Reinforcement Learning (RL) struggles with instability in high-throughput data streams.
To resolve this, we propose \textbf{Dual-Rerank}, a unified framework designed for industrial reranking that bridges the structural gap via \textit{Sequential Knowledge Distillation} and addresses the optimization gap using \textit{List-wise Decoupled Reranking Optimization (LDRO)} for stable online RL. Extensive A/B testing on production traffic demonstrates that Dual-Rerank achieves State-of-the-Art performance, significantly improving User satisfaction and Watch Time while drastically reducing inference latency compared to AR baselines.

\end{abstract}






\maketitle


\section{Introduction}
In Kuaishou's short-video search, users initiate queries via a search box to access a dual-column feed~\cite{dualgr, kuailive}, where the initial viewport typically displays four videos. This compact layout positions the reranking stage as a critical determinant of user satisfaction. Traditional Score-and-Sort paradigms~\cite{wide-deep, youtube, prm, dcn, list-cvae} decouple scoring from ordering; while efficient, they fail to capture dynamic inter-item dependencies, often missing optimal permutations that maximize whole-page utility. To bridge this gap, Generative Reranking has emerged as a superior paradigm by directly modeling the joint probability $P(Y|X)$. Specifically, Autoregressive (AR) models~\cite{gref, seq2slate} explicitly capture sequential dependencies, aligning with the user's browsing experience, but suffer from prohibitive $O(L)$ latency. Conversely, Non-Autoregressive (NAR) models offer $O(1)$ parallel generation but are limited by the Conditional Independence Assumption. Prior attempts like NAR4Rec~\cite{nar4rec} use inference-time patches (e.g., Contrastive Decoding) but fail to resolve the intrinsic lack of sequential modeling parameters.

\begin{figure}[t]
  \centering
  \includegraphics[width=0.95\linewidth]{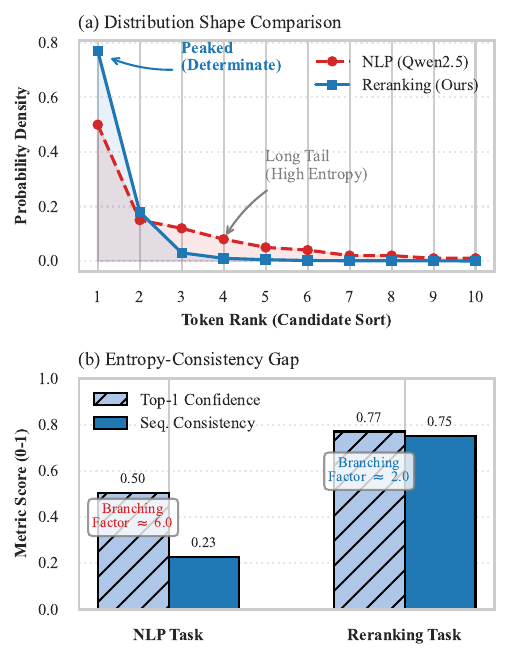}
  \vspace{-0.2cm}
  \caption{Comparison of Distribution Characteristics between NLP (Qwen2.5) and Industrial Reranking Tasks. (a) \textbf{Distribution Shape Analysis}: The curves plot the probability of the top-$k$ tokens at each decoding step, averaged over 100 samples from the Alpaca dataset (NLP) and one day of production logs (Reranking). Unlike the high-entropy, long-tail distribution observed in NLP (Red), reranking exhibits a sharp \textbf{Unimodal Concentration} (Blue). (b) Quantitative metrics reveal a significant ``Entropy-Consistency Gap,'' validating the quasi-deterministic nature of ranking. Detailed explanation of the metrics are showed in Appendix~\ref{app:metrics}.}
  \label{fig:unimodal_analysis}
\end{figure}

Instead of relying on extrinsic inference heuristics, we contend that the failure of NAR in open-ended text generation does not necessarily apply to reranking. In NLP, the core bottleneck to use NAR model is the ``Multi-Modality Problem''~\cite{nar_nlp}, where high entropy leads to flat, multi-peaked distributions. Does reranking face the same dilemma? To answer this, we conducted a rigorous comparison between a general LLM (Qwen2.5-1.5B) and an industrial AR reranker, as visualized in Figure~\ref{fig:unimodal_analysis}. 

The results reveal a critical Entropy-Consistency Gap:
1) \textit{NLP Divergence:} Due to linguistic ambiguity, the LLM exhibits a long-tail distribution (Fig.~\ref{fig:unimodal_analysis}a) with a low sequence consistency of 22.6\%, indicating that the model oscillates between multiple valid expressions.
2) \textit{Reranking Convergence:} In contrast, constrained by business utility, the reranking model demonstrates extreme stability—a sharp Unimodal Concentration with a Top-1 confidence of 0.77 and consistency reaching 75.0\%.
This empirical evidence suggests that the optimal solution space for reranking is highly constrained and peaked. Under such ``quasi-deterministic'' regimes, the complex conditional probability $P(y_t | y_{<t}, x)$ of an AR model mathematically degenerates towards a marginal distribution $P(y_t | x)$. This implies that if properly guided, a NAR model can effectively internalize the ordering patterns of an AR teacher without suffering from the multi-modal collapse common in NLP.

However, solving the structural efficiency dilemma is only half the battle. Industrial systems also face an \textit{Optimization Duality}: Supervised Learning (SL) ensures training stability but aligns only with past exposure bias, whereas Reinforcement Learning (RL) optimizes for future ecosystem utility (e.g., retention) but suffers from severe cold-start instability. To reconcile these orthogonal conflicts, we propose Dual-Rerank, a unified framework designed to harmonize both the \textbf{Structural Duality} (AR vs. NAR) and the \textbf{Optimization Duality} (Imitation vs. Evolution).

The framework unfolds in a progressive evolutionary paradigm. \textbf{In the training phase}, we design a two-stage mechanism to resolve the aforementioned dualities. First, to bridge the \textit{Structural Gap}, we leverage the \textit{Unimodal Concentration Hypothesis} and treat the AR model as a "Generative Anchor." Through Sequential Knowledge Distillation, the dependency logic of the teacher is transferred into a parallel NAR student, providing a mathematically stable initialization without incurring $O(L)$ latency. Second, to bridge the \textit{Optimization Gap}, we introduce List-wise Decoupled Reranking Optimization (LDRO). Moving beyond mere imitation, LDRO optimizes for \textit{Whole-Page Utility} by employing \textit{Vectorized Gumbel-Max} for efficient exploration and \textit{Streaming Double-Decoupling} to neutralize reward drifts in dynamic industrial streams.

\textbf{In the inference phase}, we capitalize on the computational surplus gained from the NAR architecture. By converting the "saved latency" into "search breadth," we implement a Sample-and-Rank (Best-of-N) strategy. This allows the model to approximate global optimality within strict industrial time constraints ($<20$ms), effectively completing the transition from efficiency to efficacy.

Our main contributions are summarized as follows:
\begin{enumerate}
    \item \textbf{Theoretical Insight:} We propose the \textit{Unimodal Concentration Hypothesis} and provide the $\epsilon$-Bound Ranking Stability Theorem, theoretically justifying the feasibility of transferring sequential knowledge from AR to NAR in constrained reranking tasks.
    \item \textbf{Framework \& Algorithm:} We introduce Dual-Rerank and the LDRO algorithm. LDRO innovatively resolves the "Reward Shift" and "Position Insensitivity" dilemmas inherent in applying on-policy RL to industrial streams via Double-Decoupling and Rank-Decay Modulation.
    \item \textbf{Industrial Impact:} Extensive offline and online A/B testing demonstrates that Dual-Rerank significantly improves core business metrics and reduce P99 latency, successfully validating the paradigm shift from discriminative to generative reranking in a large-scale production system.
\end{enumerate}


\section{Related Work}

\subsection{Evolution of Reranking Paradigms}
Reranking serves as the final decision layer in recommender systems, aiming to optimize the permutation of candidate items for maximum user satisfaction~\cite{prm, revisit, enhancing, variation, sort_gen}. 
Early approaches followed the \textit{Discriminative (Score-and-Sort)} paradigm. Models like PRM~\cite{prm} and DLCM~\cite{dlcm} employ attention mechanisms~\cite{transformer} or RNNs~\cite{rnn, seq2seq} to encode feature interactions. Despite incorporating local context, these methods typically assign independent scalar scores to items and sort them greedily. This approach often fails to capture the combinatorial nature of list utility, ignoring ``bundle effects'' where an item's value depends heavily on its neighbors~\cite{seq2slate, context_aware}.

To address this, \textit{Generative Reranking} formulates the task as a sequence generation problem, modeling the joint probability of the permutation to optimize holistic utility (e.g., diversity and coverage)~\cite{goalrank, gao2025llm4rerank, nlgr, gref, nar4rec}. 
Pioneering works like Seq2Slate~\cite{seq2slate} demonstrated that directly generating target indices outperforms greedy baselines. 
Subsequent research expanded this into two-stage frameworks (Generator-Evaluator), such as PIER~\cite{pier} and PRS~\cite{revisit}, which sample multiple candidate slates and select the best one based on whole-page metrics. 
Recent advances have further integrated list-wise optimization strategies, including discrete diffusion processes~\cite{diffusion} and reinforcement learning (RL) objectives~\cite{globally_r, ie2019slateq}, to better align generation with long-term user engagement. However, these complex two-stage or iterative methods often incur high computational costs, challenging their deployment in latency-sensitive industrial systems.

\subsection{Autoregressive vs. Non-Autoregressive Models}

Within the generative landscape~\cite{grn, mge, gref, nar4rec}, the trade-off between accuracy and latency is dictated by the decoding architecture.

\textbf{Autoregressive (AR) Models}, adopted by Seq2Slate~\cite{seq2slate} and recent LLM-based rerankers~\cite{llm4pr, llm4rerank, enhancing, reranking_review}, strictly follow the probability chain rule $P(Y|X)=\prod P(y_t|y_{<t}, X)$. 
By explicitly conditioning on previous selections, AR models effectively capture sequential dependencies and causal user behaviors. 
While achieving high ranking accuracy, their serial decoding nature results in $O(L)$ latency. Even efficient variants like GReF~\cite{gref} face bottlenecks under high concurrency.
Moreover, AR models trained via teacher forcing may suffer from exposure bias, leading to error propagation during inference~\cite{gref, enhancing}.

\textbf{Non-Autoregressive (NAR) Models} attempt to break this latency bottleneck by assuming conditional independence among item positions, allowing for $O(1)$ parallel generation~\cite{dd1, dd2, flow_matching}. 
Works like NAR4Rec~\cite{nar4rec} and SetRank~\cite{setrank} predict item positions or permutation matrices simultaneously. 
While significantly faster, the \textit{Conditional Independence Assumption} often leads to the ``multi-modality problem,'' resulting in incoherent lists or conflicting recommendations. 
Although techniques like contrastive decoding~\cite{nar4rec} attempt to mitigate this, they lack explicit sequential guidance. 
Different from these approaches, our Dual-Rerank framework bridges this structural gap by distilling sequential dependencies from an AR teacher into a parallel NAR student, unified with a decoupled RL optimization (LDRO) to ensure both structural consistency and business objective alignment.

\section{Methodology}

\begin{figure*}[!thbp]
  \centering
  \includegraphics[width=1\linewidth]{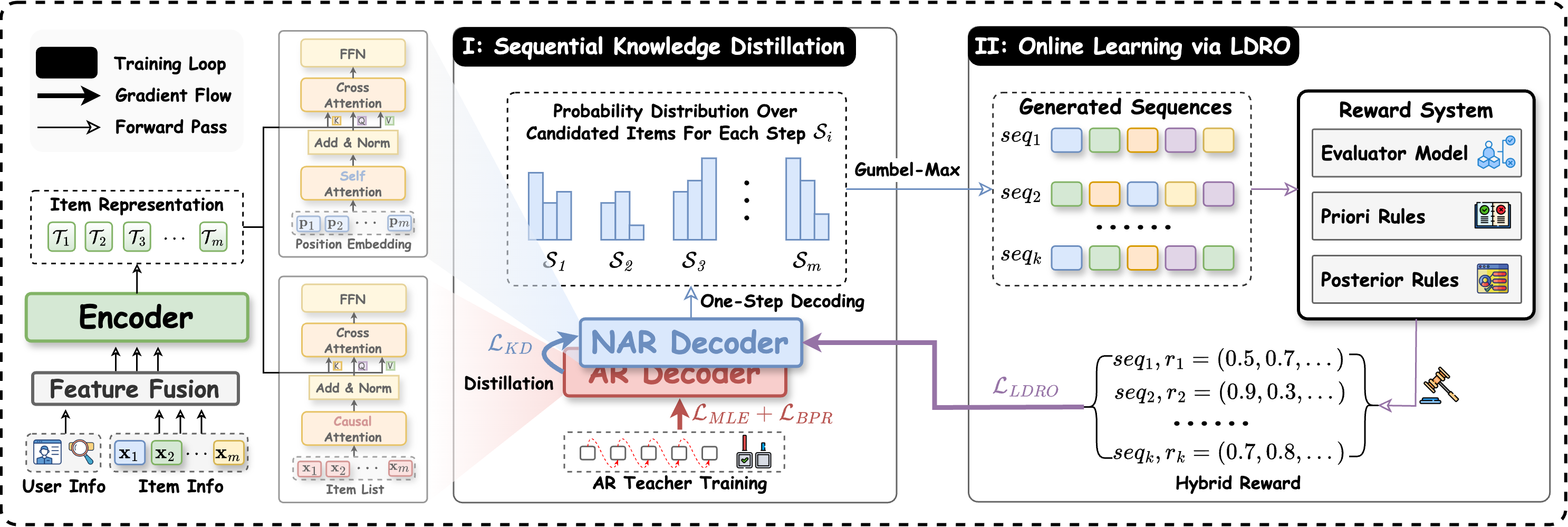}
    \caption{Overview of \textbf{Dual-Rerank}: joint online updates of an autoregressive Teacher and a non-autoregressive Student via (i) sequential distillation and (ii) online RL optimization (LDRO) under strict latency.}
  \label{fig:framework_overview}
  \vspace{-0.5em}
\end{figure*}

In this section, We propose \textbf{Dual-Rerank}, a unified online framework that reconciles the accuracy--latency tension in industrial reranking by \emph{jointly} updating an AR Teacher and a NAR Student on streaming interactions (Fig.~\ref{fig:framework_overview}, Alg.~\ref{alg:dual_rerank_streaming}).

\subsection{Framework Overview}

\textbf{Problem Formulation.} Given context $x$ (e.g., user profile, query, recent interactions) and candidate set $\mathcal{C}$, reranking aims to output an ordered slate $Y=[y_1,\dots,y_L]$ that maximizes the expected \textbf{Whole-Page Utility} $R(Y)$:
\begin{equation}
\label{pf}
    Y^* = \operatorname*{arg\,max}_{Y \in \mathcal{P}(\mathcal{C})} \mathbb{E} [R(Y) | x],
\end{equation}
where $\mathcal{P}(\mathcal{C})$ is the set of all permutations.

\textbf{Architecture.} Dual-Rerank uses two parallel branches with a shared encoder: (i) an AR \emph{Teacher} as a sequential anchor $P_{AR}(Y|x)$ trained from exposure logs and preference signals, and (ii) a NAR \emph{Student} $P_{NAR}(Y|x)$ initialized via distillation and improved via online RL for whole-page utility. We next detail (A) distillation (Sec.~\ref{sec:teacher_student}) and (B) LDRO (Sec.~\ref{sec:phase2}), followed by inference (Sec.~\ref{sec:inference}).

\begin{algorithm}[t]
\caption{Dual-Rerank Online Streaming Training Algorithm}
\label{alg:dual_rerank_streaming}
\begin{algorithmic}[1]
\Require 
    \Statex Online Interaction Stream $\mathcal{S}$; 
    \Statex Teacher Model $\pi_{\theta_{T}}$ (AR);
    \Statex Student Model $\pi_{\theta_{S}}$ (NAR);
    \Statex Loss Weights $\lambda_{KD}, \lambda_{RL}$.
    
\State Initialize parameters $\theta_{T}$ and $\theta_{S}$
\State \textbf{Input:} Continuous stream of user interaction batches $B_t = \{(x, y_{user}, r)\}_t$

\While{Stream $\mathcal{S}$ is active}
    \State Receive latest batch $B_t$ immediately after user actions
    
    \Statex \textbf{\textit{// 1. Teacher Learning (AR Path)}}
    \State Compute AR logits via sequential decoding
    \State $\mathcal{L}_{MLE} \leftarrow$ Max Likelihood Estimation on exposure logs
    \State $\mathcal{L}_{BPR} \leftarrow$ Pairwise Ranking Loss on feedback
    \State $\mathcal{L}_{Teacher} = \mathcal{L}_{MLE} + \mathcal{L}_{BPR}$

    \Statex \textbf{\textit{// 2. Sequential Distillation (Bridge)}}
    \State Compute NAR logits via parallel decoding
    \State $\mathcal{L}_{KD} \leftarrow D_{KL}(\text{StopGrad}(\pi_{\theta_{T}}) || \pi_{\theta_{S}})$ 
    \State \textit{Note: Teacher serves as a dynamic soft-label generator.}

    \Statex \textbf{\textit{// 3. Online Exploration (RL Path)}}
    \State Generate $K$ candidates via Gumbel-Max using $\pi_{\theta_{S}}$
    \State Compute Hybrid Rewards $R$ \& Apply Double-Decoupling
    \State $\mathcal{L}_{LDRO} \leftarrow$ Policy Gradient with Rank-Decay Modulation

    \Statex \textbf{\textit{// 4. Unified Gradient Update}}
    \State $\mathcal{L}_{Total} = \mathcal{L}_{Teacher} + \lambda_{KD} \cdot \mathcal{L}_{KD} + \lambda_{RL} \cdot \mathcal{L}_{LDRO}$
    \State Compute Gradients $\nabla \mathcal{L}_{Total}$
    \State Update $\theta_{T}$ and $\theta_{S}$ simultaneously via optimizer
\EndWhile

\end{algorithmic}
\end{algorithm}

\subsection{Sequential Knowledge Distillation}
\label{sec:teacher_student}

This component aims to endow the structurally efficient NAR Student with the sequential dependency modeling capabilities of the AR Teacher. In our joint streaming framework, this process functions as a dynamic "bridge," continuously transferring the teacher's evolving ranking logic to the student.

\subsubsection{The Sequential Anchor: Preference-Aware AR Teacher}
We adopt a Pointer Network-based AR Teacher and train it with a hybrid objective that captures both exposure bias and preference.


\textbf{1) Exposure Learning via MLE.} Given the exposed sequence $Y=\{y_1,\dots,y_L\}$:
\begin{equation}
\label{eq:mle}
    \mathcal{L}_{MLE} = - \sum_{t=1}^{L} \log \pi_{\theta_T}(y_t \mid y_{<t}, x).
\end{equation}

\textbf{2) Preference Injection via BPR.} We define an item utility score from feedback:
\begin{equation}
\label{eq:score}
    S(y) = 1.0 \cdot \mathbb{I}(\text{click}) + 2.0 \cdot \mathbb{I}(\text{long\_view}) + 0.1 \cdot \mathbb{I}(\text{exposure}),
\end{equation}
and construct $\mathcal{D}_{pair}=\{(y_i,y_j)\mid S(y_i)-S(y_j)>\delta_{score}\}$. The BPR loss is:
\begin{equation}
\label{eq:bpr}
    \mathcal{L}_{BPR} = - \sum_{(y_i, y_j) \in \mathcal{D}_{pair}} \ln \sigma\!\left(s_{\theta_T}(y_i|x) - s_{\theta_T}(y_j|x)\right),
\end{equation}
where $s_{\theta_T}$ denotes AR decoder logits.

\textbf{3) Total AR Objective.}
\begin{equation}
\label{eq:teacher_total}
    \mathcal{L}_{Teacher} = \mathcal{L}_{MLE} + \lambda_{bpr}\mathcal{L}_{BPR}.
\end{equation}

\subsubsection{Distillation and Feasibility}

With a robust AR Teacher, the challenge lies in transferring its sequential wisdom to the parallel Student. In open-ended generation tasks (e.g., NLP), NAR models often suffer from the "Multi-Modality Problem", leading to incoherent outputs.
However, we argue that the \textit{Reranking} task possesses a unique property that fundamentally distinguishes it from open-ended generation. We formalize this insight as follows:

\newtheorem{hypothesis}{Hypothesis}
\begin{hypothesis}[Contextual Determinism]
Given a sufficiently expressive encoder context $x$ (comprising user profile, interaction history, and item features), the optimal ranking $Y$ is quasi-deterministic. The conditional entropy $H(Y|x)$ approaches zero.
\end{hypothesis}

Under Hypothesis 1, knowing the full context $x$ provides sufficient information to determine the rank of item $y_t$. Mathematically, this implies that the information gain provided by the history $y_{<t}$ is negligible given $x$ (i.e., $I(y_t; y_{<t} | x) \approx 0$).
Consequently, the complex conditional distribution $P(y_t | y_{<t}, x)$ degenerates towards the marginal $P(y_t | x)$.
This justifies the use of Sequence-Level Knowledge Distillation to force the NAR Student $\pi_{\theta_S}$ to fit the modes of the AR Teacher:
\begin{equation}
\label{eq:distill_loss}
    \mathcal{L}_{KD} = \sum_{t=1}^{L} D_{KL}(\text{stop\_grad}(\pi_{\theta_T}(\cdot|y_{<t}, x)) \parallel \pi_{\theta_S}(\cdot|x))
\end{equation}
where `stop\_grad' ensures the Teacher is not destabilized by the Student's gradients.

We further provide a stability bound for ranking flips (proof in Appendix~\ref{sec:appendix_theory}):

\begin{theorem}[$\epsilon$-Bound Ranking Stability]
Given an AR Teacher with a robust ranking margin $\delta$ (confidence), and a distillation error $\epsilon$ (measured by Total Variation Distance), the probability of the NAR Student producing a \textbf{Ranking Flip} (i.e., erroneously predicting $P(y_j) > P(y_i)$ when the teacher indicates $y_i \succ y_j$) is strictly upper-bounded by:
\begin{equation}
    P(\text{Flip}) \le \frac{\sqrt{2\epsilon}}{\delta}
\end{equation}
where $\xi$ is a slack variable related to the teacher's uncertainty.
\end{theorem}

\subsubsection{Empirical Validation: Position-wise Top-1 Agreement Rate}
\label{sec:empirical_validation}

We validate the approximation implied by Hypothesis 1 via the \textbf{Position-wise Top-1 Agreement Rate (PTAR)} between the Teacher (history-conditioned) and Student (context-only).

\textbf{Metric.} For context $x$ and teacher condition $y_{<t}$ (from ground truth), define:
\begin{equation}
\label{eq:ptar_argmax}
    \hat{y}_t^{T} = \arg\max_{v \in \mathcal{V}} \pi_{\theta_T}(v \mid y_{<t}, x), \quad
    \hat{y}_t^{S} = \arg\max_{v \in \mathcal{V}} \pi_{\theta_S}(v \mid x).
\end{equation}
PTAR averages the position-wise match rate over $\mathcal{D}$:
\begin{equation}
\label{eq:ptar}
    \text{PTAR} = \frac{1}{|\mathcal{D}| \cdot L} \sum_{(x, Y) \in \mathcal{D}} \sum_{t=1}^{L} \mathbb{I}(\hat{y}_t^{T} = \hat{y}_t^{S}).
\end{equation}

\textbf{Rationale for Validation.}
The AR Teacher makes decisions based on the \textit{full information set} (Context + History), while the NAR Student relies \textit{solely on Context}.
1) If the sequential dependency $y_{<t}$ were strictly indispensable for correctness (i.e., $y_t$ changes depending on $y_{<t}$), the Context-only Student would theoretically lack the necessary information to match the Teacher's predictions, leading to a low PTAR.
2) Conversely, a high PTAR empirically proves that the optimal decision $\hat{y}_t$ is effectively determined by $x$ alone. In this case, the complex sequential reasoning of the Teacher collapses into a \textit{quasi-deterministic} outcome that can be accurately "cloned" by the Student, validating our approximation assumption.

\textbf{Result Analysis.} Figure~\ref{fig:agreement_rate} compares the training dynamics of a Pure NAR versus our Distilled NAR.
Distillation yields higher and smoother PTAR than a non-distilled NAR baseline, supporting that the Student can reliably approximate the Teacher’s position-wise decisions from context alone.

\begin{figure}[t]
    \centering
    \includegraphics[width=0.9\linewidth]{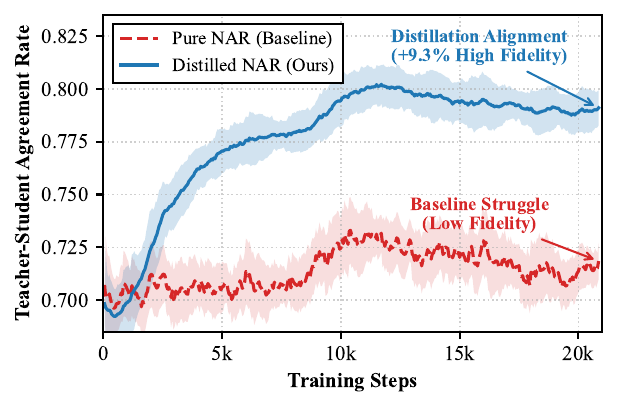}
    \vspace{-0.5em}
    \caption{Teacher--Student PTAR during training. Distillation improves agreement and stability over a pure NAR baseline.}
    \label{fig:agreement_rate}
    \vspace{-1.5em}
\end{figure}

Distillation provides a stable structural prior, but optimizing \textbf{Whole-Page Utility} requires online learning with delayed, multi-objective, and non-stationary feedback. We therefore optimize the Student with \textbf{List-wise Decoupled Reranking Optimization (LDRO)}.

Applying RL in a high-throughput industrial stream presents three specific challenges:
\textbf{(1) Sampling Efficiency Bottleneck:} On-Policy RL requires extensive online exploration. Traditional AR-based sampling incurs $O(L)$ latency, creating a severe bottleneck for real-time training throughput.
\textbf{(2) Streaming Distribution Drift:} Unlike stationary batch training, streaming data suffers from "Reward Shift" (e.g., CTR fluctuates between peak and off-peak hours). Standard RL baselines often fail to adapt to these non-stationary baselines, leading to gradient instability.
\textbf{(3) Position Sensitivity Mismatch:} Standard RL objectives (maximizing cumulative reward) typically imply a linear contribution of actions. This contradicts the \textbf{Cascading Attention Model} in search, where an error at Rank 1 causes exponentially higher utility loss than at Rank 10.

To systematically address these issues, LDRO integrates three tailored components.

\subsubsection{Challenge 1 Solution: High-Throughput Parallel Exploration}
To resolve the sampling bottleneck, we capitalize on the parallel nature of the \textbf{NAR architecture}. By combining NAR with the \textbf{Vectorized Gumbel-Max} trick, we transform the stochastic sampling process into fully vectorized tensor operations (Addition + Argmax).
For a position $t$ and item $i$, the sampled action is derived as:
\begin{equation}
\label{eq:gumbel}
    y_t = \arg\max_{i} \left( \frac{\log \pi_{\theta_S}(i|x) + \mathcal{G}_{i}}{\tau} \right)
\end{equation}
where $\mathcal{G} \sim \text{Gumbel}(0, 1)$. Crucially, this allows us to generate $K$ candidate slates $\{Y^{(1)}, \dots, Y^{(K)}\}$ in $O(1)$ time, providing the massive data throughput required for stable On-Policy RL updates.

\begin{figure}[!t]
    \centering
    \includegraphics[width=0.9\linewidth]{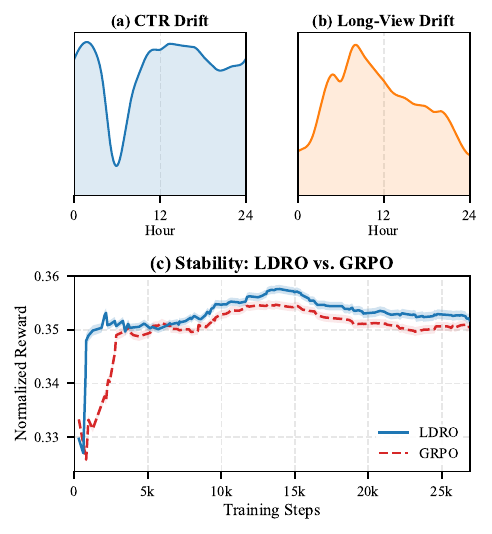} 
    \vspace{-0.5em}
    \caption{Stability under streaming drift. (a) \& (b): Real-world logs reveal significant drifts in different time. (c): Comparison of Total Reward during training.}
    \label{fig:drift_stability}
    \vspace{-1em}
\end{figure}

\subsubsection{Challenge 2 Solution: Online Double-Decoupling}
Given $K$ generated slates, the core challenge lies in constructing a stable reward signal from a non-stationary data stream.
we conduct a Hybrid Reward and apply a Double-Decoupling Normalization to further stabilize the learning process.

\paragraph{1. Hybrid Reward Construction.}
To capture the multidimensional "Whole-Page Utility," we construct a reward vector $\mathbf{r} \in \mathbb{R}^{M}$ from three sources:
\textbf{(1) Reward Network ($R_{net}$):} Since online feedback is delayed, we employ a pre-trained multi-task network to estimate immediate outcomes and long-term satisfaction.
\textbf{(2) Priori Constraints ($R_{prior}$):} To adhere to business layouts, we compute relevance NDCG explicitly for the visible windows (Top-4 and Top-8).
\textbf{(3) Posterior Feedback ($R_{post}$):} We utilize real-time user actions (Clicks/Long-Views) as sparse but ground-truth signals to correct the bias of the Reward Network.

\paragraph{2. Online Double-Decoupling.}
To address distribution shifts. LDRO applies a two-stage decoupling process: 

\begin{figure*}[!t]
    \centering
    \includegraphics[width=\linewidth]{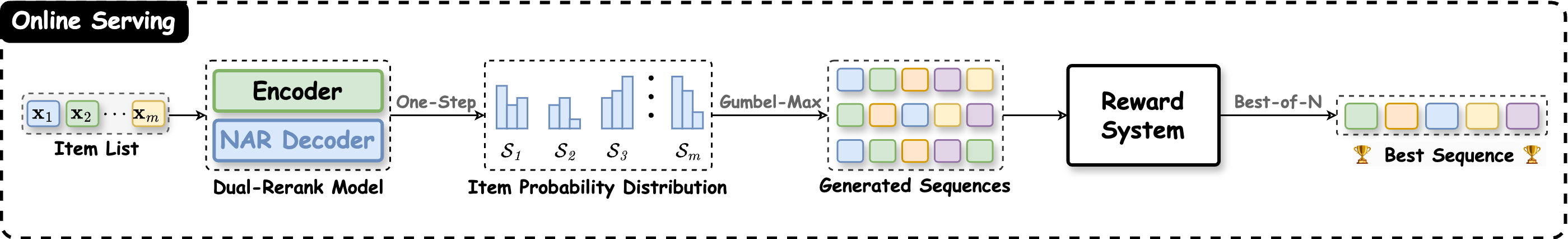} 
    \caption{Overview of the Online Serving Phase. The framework leverages One-Step NAR decoding and Vectorized Gumbel-Max sampling to generate parallel candidates, enabling a high-throughput Best-of-N selection strategy.}
    \label{fig:online_serving}
\end{figure*}

\textbf{Step 1: Group-wise Normalization.}
For a fixed context $x$, we sample $K$ candidate slates $\{Y^{(1)},\dots,Y^{(K)}\}$ from the current policy. 
We compute a \textbf{group advantage} by normalizing objective-$m$ rewards within this $K$-slate set:
\begin{equation}
    A_{m}^{(j)} = \frac{r_{m}^{(j)} - \mu_{group, m}}{\sigma_{group, m} + \epsilon}
\end{equation}

\textbf{Step 2: Batch-wise Standardization.}
While Step 1 handles query-level variance, it cannot handle \textbf{Global Distribution Drift}. As illustrated in Figure~\ref{fig:drift_stability} (a) and (b), real-world traffic exhibits significant "Tidal Drifts" (e.g., peaks vs. valleys in CTR and Long-View rates).
To prevent gradient explosions caused by these shifts, we apply a second normalization across the current \textit{mini-batch} $\mathcal{B}$:
\begin{equation}
    \hat{A}_{m}^{(j)} = \frac{A_{m}^{(j)} - \mu_{\mathcal{B}, m}}{\sigma_{\mathcal{B}, m} + \epsilon}
\end{equation}
This batch-wise standardization effectively neutralizes the non-stationary nature of the stream. 
The impact of this stabilization is evident in \textbf{Figure~\ref{fig:drift_stability} (c)}. Compared to the standard Group Relative Policy Optimization (GRPO) baseline, which suffers from reward volatility (Red Curve) due to the aforementioned drifts, LDRO (Blue Curve) achieves significantly faster convergence and maintains a consistently higher normalized reward throughout the streaming training phase.
Finally, the multi-objective advantages are fused: $A_{total}^{(j)} = \sum \alpha_m \hat{A}_{m}^{(j)}$.

\subsubsection{Challenge 3 Solution: Rank-Weighted Policy Optimization}
To align optimization with user's \textbf{Cascading Attention}, we introduce \textbf{View-Port Aware Rank Modulation}. Acknowledging that attention remains consistently high in the initial "First Screen" before decaying, we apply a generic \textbf{Piecewise Decay Function}:

\begin{equation}
    w_t = 
    \begin{cases} 
    1.0, & \text{if } 1 \le t \le K \quad \text{(Prime Visibility Zone)} \\
    \frac{1}{\log_2(t - K + 2)}, & \text{if } t > K \quad \text{(Tail Zone)}
    \end{cases}
\end{equation}

where $K$ denotes the effective viewport capacity. Given Kuaishou's typical \textbf{dual-column layout}~\cite{edge, kuailive, dualgr}, the first screen typically displays a $2 \times 2$ grid, which leads us to set $K=4$. This prioritizes high-impression items in the prime zone while enabling structured exploration in the tail.

\paragraph{The Unified Objective.}
Combining the Structural Prior (Distillation) and the Rank-Weighted Exploration (RL), the final objective for the Student is:
\begin{equation}
\label{eq:final_loss}
\begin{aligned}
    \mathcal{L}_{Total} = & {\mathcal{L}_{Teacher}} + \lambda_{KD} \cdot {\mathcal{L}_{KD}} \\
     & - \lambda_{RL} \cdot \frac{1}{G} \sum_{j=1}^{G} A_{total}^{(j)} \sum_{t=1}^{L} w_t \log \pi_{\theta_S}(y_{t}^{(j)}|x)
%
\end{aligned}
\end{equation}
This unified formulation allows Dual-Rerank to co-evolve via joint streaming updates.

\subsection{Inference: Sample-and-Rank Strategy}
\label{sec:inference}

As illustrated in Figure~\ref{fig:online_serving}, the online inference phase is designed to convert the structural efficiency of the NAR architecture into ranking capacity.
Unlike Autoregressive models that suffer from $O(L)$ serial decoding latency, Dual-Rerank executes a streamlined \textbf{Sample-and-Rank} pipeline:
\textbf{(1) One-Step Decoding:} The Encoder first processes the context $x$. Subsequently, the NAR Decoder predicts the item probability distribution $\pi_{\theta_S}(\cdot|x)$ for all ranking positions simultaneously via a single forward pass ($O(1)$).
\textbf{(2) Vectorized Sampling:} To explore the combinatorial solution space, we employ Vectorized Gumbel-Max sampling to generate $N$ diverse candidate slates in parallel.
\textbf{(3) Reward Scoring:} These candidates are immediately scored by the inference-time Reward System (consisting of the Reward Network $R_{net}$ and Priori Rules $R_{prior}$). Note that Posterior feedback is unavailable during inference, so we rely on the generalized utility learned by $R_{net}$.
\textbf{(4) Best-of-N Selection:} Finally, the sequence with the highest estimated score is selected as the Best Sequence for exposure.

This strategy effectively transforms the computational surplus of NAR (saved from avoiding serial decoding) into a significant gain in Search Breadth. It approximates global optimality within strict industrial latency constraints.


\section{Experiments}
\label{sec:experiments}

To comprehensively evaluate the effectiveness of Dual-Rerank, we conduct experiments to address the following research questions:
\begin{itemize}
    \item \textbf{RQ1 (Overall Effectiveness):} Does Dual-Rerank break the accuracy-latency trade-off, outperforming state-of-the-art NAR baselines while achieving parity with computationally expensive AR models?
    \item \textbf{RQ2 (Structural Fidelity):} Can the proposed \textit{Sequential Knowledge Distillation} (Phase 1) effectively transfer the dependency logic from AR to NAR, preventing the ``Ranking Collapse'' often observed in parallel generation?
    \item \textbf{RQ3 (Optimization Stability):} Does the \textit{LDRO} algorithm (Phase 2) effectively handle the non-stationary reward shifts in industrial streaming environments?
    \item \textbf{RQ4 (Industrial Viability):} How does the framework perform in a large-scale production environment in terms of system latency and core business metrics?
\end{itemize}

\subsection{Experimental Setup}

\subsubsection{Datasets}
We utilize a public benchmark and a massive industrial dataset to ensure reproducibility and practical applicability.

\begin{itemize}
    \item \textbf{Avito Context Ad Clicks}\footnote{https://www.kaggle.com/c/avito-context-ad-clicks/data}: A widely used benchmark containing over 53 million search sessions, 1.3 million users, and 36 million ads. Following standard protocols~\cite{nar4rec, gref}, we use data from the first 21 days for training and the subsequent 7 days for testing. The task involves predicting optimal permutations for lists of length 5 based on context features.

    \item \textbf{Kuaishou Production Dataset}: Collected from the Kuaishou App (400M+ DAUs), this dataset validates industrial robustness with \textbf{1 billion ($10^9$)} interaction logs. Each sample includes user profiles, 30 candidate items, and the top-10 exposed list. Labels incorporate multi-objective signals (Clicks, Long-Views et al.).
\end{itemize}

\subsubsection{Baselines}
We compare Dual-Rerank against a comprehensive set of state-of-the-art baselines, categorized into Discriminative (Pointwise/Listwise) and Generative (AR/NAR) paradigms.

\textbf{1. Discriminative Methods}
These models score items either independently or collectively.
\textbf{(1) DNN}~\cite{youtube}: Uses Multi-Layer Perceptrons (MLP) to estimate pointwise click-through rates (CTR) independently.
\textbf{(2) DCN}~\cite{dcn}: Enhances DNNs by introducing cross networks to explicitly learn bounded-degree feature interactions.
\textbf{(3) PRM}~\cite{prm}: A listwise approach employing self-attention to capture inter-item correlations and global context for scoring.

\textbf{2. Generative Methods (AR \& NAR)}
These methods treat reranking as sequence generation or generative evaluation.
\textbf{(1) Seq2Slate}~\cite{seq2slate}: An Autoregressive (AR) baseline using Pointer Networks for sequential selection, representing the accuracy upper bound for serial decoding.
\textbf{(2) Edge-Rerank}~\cite{edge}: A mobile-oriented framework generating sequences via adaptive beam search over pointwise scores.
\textbf{(3) PIER}~\cite{pier}: A two-stage framework combining hashing-based candidate selection with joint generator-evaluator training.
\textbf{(4) NAR4Rec}~\cite{nar4rec}: The SOTA Non-Autoregressive (NAR) model utilizing parallel decoding with unlikelihood training.
\textbf{(5) MG-E}~\cite{mge}: Deploys diverse generators with a ``list comprehensiveness'' metric to propose complementary lists for evaluation.

\subsection{Overall Effectiveness (RQ1)} 

We present the comparative results in Table~\ref{tab:overall_performance}. 

\begin{table}[t]
    \caption{Overall performance comparison on Avito and Kuaishou datasets. The best results are highlighted in \textbf{bold}, and the second-best are \underline{underlined}. Improvement is calculated relative to the best baseline (GReF).} 
    \label{tab:overall_performance}
    \centering
    \begin{tabular}{l | c c | c c}
        \toprule
        \multirow{2}{*}{\textbf{Method}} & \multicolumn{2}{c|}{\textbf{Avito (Public)}} & \multicolumn{2}{c}{\textbf{Kuaishou (Industrial)}} \\
        \cline{2-5}
         & \textbf{AUC} & \textbf{NDCG} & \textbf{AUC} & \textbf{NDCG} \\
        \midrule
        \multicolumn{5}{l}{\textit{Discriminative Methods}} \\
        DNN~\cite{youtube} & 0.6614 & 0.6920 & 0.6866 & 0.7122 \\
        DCN~\cite{dcn} & 0.6623 & 0.7004 & 0.6879 & 0.7116 \\
        PRM~\cite{prm} & 0.6881 & 0.7380 & 0.7119 & 0.7396 \\
        Edge-rerank~\cite{edge} & 0.6953 & 0.7203 & 0.7143 & 0.7401 \\
        PIER~\cite{pier} & 0.7109 & 0.7401 & 0.7191 & 0.7387 \\
        \midrule
        \multicolumn{5}{l}{\textit{Generative Methods (NAR)}} \\
        NAR4Rec~\cite{nar4rec} & 0.7234 & 0.7409 & 0.7254 & 0.7425 \\
        MG-E~\cite{mge} & 0.7323 & 0.7437 & 0.7381 & 0.7453 \\
        \midrule
        \multicolumn{5}{l}{\textit{Generative Methods (AR)}} \\
        Seq2Slate~\cite{seq2slate} & 0.7034 & 0.7225 & 0.7165 & 0.7383 \\
        GReF~\cite{gref} & \underline{0.7384} & \underline{0.7478} & \underline{0.7387} & \underline{0.7498} \\
        \midrule
        \textbf{Dual-Rerank} & \textbf{0.7441} & \textbf{0.7517} & \textbf{0.7448} & \textbf{0.7565} \\
        \emph{Improvement} & \emph{+0.77\%} & \emph{+0.52\%} & \emph{+0.83\%} & \emph{+0.89\%} \\
        \bottomrule
    \end{tabular}
\end{table}

\paragraph{Surpassing the NAR Bottleneck}
As shown in Table~\ref{tab:overall_performance}, Dual-Rerank significantly outperforms the SOTA NAR baseline, NAR4Rec, by a substantial margin (e.g., \textbf{+2.6\%} relative AUC gain on Kuaishou). This result offers strong empirical support for our distillation strategy: by explicitly injecting sequential dependency logic during Phase 1, we successfully overcome the ``Conditional Independence'' bottleneck that typically constrains parallel decoding models.

\paragraph{Beating the AR Teacher}
Crucially, Dual-Rerank not only matches but exceeds the performance of the strongest AR baseline, GReF (\textbf{0.7448} vs. 0.7387 on Kuaishou). While AR models are limited to imitating historical ground truth (via MLE), our Phase 2 optimization (LDRO) allows the student to transcend the teacher by directly optimizing for \textit{Whole-Page Utility} (Reward). This confirms that properly guided parallel exploration can yield better global optima than greedy serial decoding.

\subsection{Structural Fidelity (RQ2)}

To investigate whether the NAR student truly internalizes the sequential logic of the AR teacher—rather than merely memorizing pointwise scores—we introduce a structural consistency metric: \textbf{Ranking Flip Rate (RFR)}.

\paragraph{Metric Definition}
Standard metrics like MSE fail to capture the combinatorial nature of ranking. RFR measures the proportion of discordant pairs between the Teacher and Student. Let $\mathcal{S}^{T}$ and $\mathcal{S}^{S}$ denote the scoring functions of the Teacher and Student. A ``flip'' occurs for a pair $(i, j)$ if the Student reverses the Teacher's preference (i.e., $\mathcal{S}^{S}_i < \mathcal{S}^{S}_j$ given $\mathcal{S}^{T}_i > \mathcal{S}^{T}_j$). RFR is defined as:
\begin{equation}
\text{RFR} = \frac{\sum_{(i, j) \in \mathcal{P}} \mathbb{I}(\mathcal{S}^{S}_i < \mathcal{S}^{S}_j \mid \mathcal{S}^{T}_i > \mathcal{S}^{T}_j)}{|\mathcal{P}|}
\end{equation}
Lower RFR indicates higher fidelity to the teacher's sequential reasoning.

\begin{figure}[!th]
    \centering
    \includegraphics[width=0.95\linewidth]{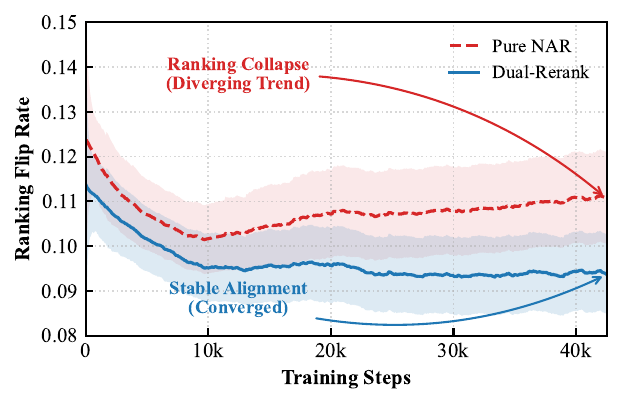}
    \caption{\textbf{Structural Fidelity Analysis (Ranking Flip Rate).} unlike PTAR which measures pointwise memorization, RFR evaluates the \textbf{pairwise ranking logic}. The divergence of the ``Pure NAR'' baseline reveals a \textit{Ranking Collapse} phenomenon (loss of global order), while Dual-Rerank maintains structural consistency.}
    \label{fig:flip_rate}
\end{figure}

\paragraph{Analysis of Training Dynamics}
Figure~\ref{fig:flip_rate} visualizes the evolution of RFR. A critical phenomenon, which we term \textbf{Ranking Collapse}, is observed in the ``Pure NAR'' baseline (Red Curve). Despite converging pointwise loss, its structural alignment with the teacher deteriorates, with RFR escalating significantly. This suggests that without explicit sequential guidance, NAR models struggle to resolve inter-item dependencies in the tail. In contrast, Dual-Rerank (Blue Curve) maintains a low and stable RFR, validating that \textit{Sequential Knowledge Distillation} effectively acts as a structural bridge, compressing the teacher's autoregressive patterns into the student's parallel parameters.

\paragraph{Impact of Contextual Depth}
We further analyze the necessity of full-sequence modeling in Table~\ref{tab:mechanism_comparison_ks}. The ``AR (Hybrid)'' strategy, which only models the top-2 items autoregressively, fails to match the performance of the full AR teacher. This indicates that dependency chains in short-video feeds cascade beyond the head. Our Distilled NAR model, however, achieves parity with the AR Teacher (\textbf{0.7391} vs. 0.7387), proving that the student has successfully captured these global dependencies.

\begin{table}[t]
    \caption{Impact of Sequential Modeling Strategies (Phase 1). * indicates parity with the AR Teacher.}
    \label{tab:mechanism_comparison_ks}
    \centering
    \resizebox{\linewidth}{!}{
    \begin{tabular}{l l l c}
        \toprule
        \textbf{Model} & \textbf{Training} & \textbf{Inference} & \textbf{AUC} \\
        \midrule
        Pure NAR & MLE & Parallel & 0.7254 \\
        \midrule
        \multicolumn{4}{l}{\textit{Autoregressive Teachers}} \\
        AR (Hybrid) & AR & Top-2 AR + Greedy & 0.7325 \\
        AR (Full) & AR & Full AR Decoding & \underline{0.7387} \\
        \midrule
        \textbf{Distilled NAR} & \textbf{Sequential Dist.} & \textbf{Parallel (NAR)} & \textbf{0.7391}$^*$ \\
        \emph{vs. Pure NAR} & & & \emph{(+1.89\%)} \\
        \bottomrule
    \end{tabular}
    }
\end{table}

\subsection{Ablation \& Training Stability (RQ3)}

To isolate the contributions of our proposed components, we conduct an ablation study on the Kuaishou dataset (Table~\ref{tab:ablation}).

\begin{table}[h]
    \caption{Ablation study of Dual-Rerank components.}
    \label{tab:ablation}
    \centering
    \begin{tabular}{l c c}
        \toprule
        \textbf{Variant} & \textbf{AUC} & \textbf{Drop} \\
        \midrule
        \textbf{Dual-Rerank (Full)} & \textbf{0.7448} & - \\
        w/ GRPO (Replace LDRO) & 0.7395 & -0.71\% \\
        w/o Distillation (Phase 1) & 0.7254 & -2.60\% \\
        w/o Double-Decoupling (Phase 2) & 0.7402 & -0.62\% \\
        w/o Rank-Weighting & 0.7429 & -0.26\% \\
        \bottomrule
    \end{tabular}
\end{table}

\paragraph{Key Observations}
\begin{enumerate}
    \item \textbf{Distillation is Foundational:} Removing Phase 1 causes the largest drop (\textbf{-2.60\%}). This confirms that the ``Sequential Anchor'' provided by the teacher is a prerequisite; without it, the model starts from a chaotic state and fails to converge.
    
    \item \textbf{Superiority of LDRO:} Replacing LDRO with standard GRPO leads to a \textbf{-0.71\%} drop. This performance regression is comparable to disabling 
    
    \item \textbf{Double-Decoupling} (-0.62\%), validating that generic RL (like GRPO) struggles with industrial ``Reward Shift,'' which our decoupling strategy effectively neutralizes.
    
    \item \textbf{Rank-Weighting Aligns Utility:} The removal of Rank-Weighting results in a \textbf{-0.26\%} drop. While smaller, this helps ensure the model prioritizes correctness in the prime visibility zone (Top-4).
\end{enumerate}

\subsection{Efficiency \& Online Impact (RQ4)}

\paragraph{System Efficiency Analysis}
Dual-Rerank is designed to maximize the \textit{Efficiency-Accuracy Pareto Frontier}. As shown in Table~\ref{tab:efficiency_comparison}, by shifting from serial $O(L)$ to parallel $O(1)$ decoding, we reduce average latency by \textbf{43.7\%} (21.5ms $\rightarrow$ 12.1ms). 
This computational surplus is strategic: it allows us to deploy the more expensive ``Best-of-N'' sampling strategy within the strict online time budget, effectively converting saved latency into expanded search breadth.

\begin{table}[h]
    \caption{System Inference Efficiency Comparison.}
    \label{tab:efficiency_comparison}
    \centering
    \begin{tabular}{l c c c}
        \toprule
        \textbf{System} & \textbf{Decoding} & \textbf{Avg. Latency} & \textbf{Speedup} \\
        \midrule
        Production AR & Serial $O(L)$ & 21.5 ms & 1.0$\times$ \\
        \textbf{Dual-Rerank} & \textbf{Parallel $O(1)$} & \textbf{12.1 ms} & \textbf{1.8$\times$} \\
        \bottomrule
    \end{tabular}
\end{table}

\paragraph{Online A/B Testing}
We deployed Dual-Rerank on the Kuaishou App, conducting a rigorous A/B test on \textbf{5\% of production traffic over a period of one month}
. 
The results (Table~\ref{tab:online_results}) show statistically significant gains.
The \textbf{+0.373\%} increase in Long-View Rate indicates that the model is optimizing for deep user engagement rather than click-bait. Simultaneously, the \textbf{-0.709\%} drop in Query Reformulation Rate suggests a substantial improvement in search relevance and user satisfaction. These metrics confirm that Dual-Rerank successfully translates theoretical advancements into tangible business value.

\begin{table}[h]
    \caption{Online A/B testing results. (* indicates $p < 0.05$).}
    \label{tab:online_results}
    \centering
    \begin{tabular}{l c}
        \toprule
        \textbf{Metric} & \textbf{Relative Improvement} \\
        \midrule
        \multicolumn{2}{l}{\textit{User Satisfaction Metrics}} \\
        Long-View Rate & \textbf{+1.107\%}$^*$ \\
        Whole-Page CTR & \textbf{+0.714\%}$^*$ \\
        Query Reformulation Rate $\downarrow$ & \textbf{-1.309\%}$^*$ \\
        \bottomrule
    \end{tabular}
\end{table}

\section{Conclusion}
We proposed Dual-Rerank to reconcile the tension between accuracy and efficiency in industrial reranking. Validating the \textbf{Unimodal Concentration Hypothesis}, we compressed AR dependencies into a parallel NAR student via \textbf{Sequential Knowledge Distillation}. Furthermore, our LDRO algorithm stabilizes streaming RL to optimize whole-page utility beyond imitation. Experiments on billion-scale datasets confirm that Dual-Rerank outperforms both SOTA NAR baselines and AR teachers. Online deployment demonstrates significant gains in user satisfaction with a 40\% latency reduction. Future work will explore scaling this paradigm to foundation models to exploit emergent list-wise reasoning abilities.

\bibliographystyle{ACM-Reference-Format}
\bibliography{sample-base}

\appendix

\appendix

\section{Theoretical Analysis}
\label{sec:appendix_theory}

In this section, we provide the rigorous mathematical proofs for the feasibility and stability of the Dual-Rerank framework proposed in the main text.

\subsection{Proof of Feasibility: Sequential Linearization}
\label{sec:proof_feasibility}

The core premise of Dual-Rerank is the \textit{Sequential Linearization Hypothesis}, which posits that explicit autoregressive (AR) dependency modeling is redundant given a sufficiently expressive encoder and a deterministic ranking task. We formalize this using Information Theory.

\subsubsection{Problem Formulation}
Let $X$ denote the input context (comprising user profile, interaction history, and candidate item features), and $Y = \{y_1, \dots, y_L\}$ denote the optimal permutation sequence.
\begin{itemize}
    \item An AR model estimates the joint probability as $P_{AR}(Y|X) = \prod_{t=1}^L P(y_t | y_{<t}, X)$.
    \item A NAR model estimates it as $P_{NAR}(Y|X) \approx \prod_{t=1}^L P(y_t | X)$.
\end{itemize}
The information loss incurred by the NAR approximation is quantified by the \textbf{Conditional Mutual Information (CMI)} between the current token $y_t$ and the history $y_{<t}$ given context $X$:
\begin{equation}
    I(y_t; y_{<t} | X) = H(y_t | X) - H(y_t | y_{<t}, X)
\end{equation}
where $H(\cdot)$ denotes the entropy.

\subsubsection{Theoretical Justification via DPI}
\textbf{Hypothesis (Unimodal Concentration).} In industrial reranking scenarios, the optimal sorting order is quasi-deterministic given the full context (i.e., utility is strictly hierarchical). Consequently, the conditional entropy of the optimal ranking approaches zero: $H(Y|X) \to 0$.

\textbf{Theorem A.1 (Vanishing Sequential Dependency).}
\textit{If the encoder $f_\theta(X)$ serves as a sufficient statistic for the optimal ranking $Y$, and the target distribution is unimodal, then the conditional mutual information $I(y_t; y_{<t} | X)$ approaches zero.}

\begin{proof}
By the chain rule of mutual information, the total dependency is the sum of step-wise dependencies:
\begin{equation}
    I(Y; X) = \sum_{t=1}^L I(y_t; X | y_{<t})
\end{equation}
Consider the Data Processing Inequality (DPI) applied to the reranking Markov chain: $Y_{truth} \to X \to \text{Encoder}(X) \to \hat{Y}$. A powerful Deep Bidirectional Encoder (e.g., BERT-based) is trained to maximize $I(Y_{truth}; \text{Encoder}(X))$.

Under the \textit{Unimodal Concentration} assumption, knowing $X$ fully determines the label $Y$. Thus, the conditional entropy $H(y_t | X) \approx 0$.
Since entropy is non-negative and conditioning cannot increase entropy:
\begin{equation}
    0 \le H(y_t | y_{<t}, X) \le H(y_t | X)
\end{equation}
Given $H(y_t | X) \approx 0$, it follows by the Sandwich Theorem that $H(y_t | y_{<t}, X) \approx 0$. Substituting this back into Eq. (16):
\begin{equation}
    I(y_t; y_{<t} | X) \approx 0 - 0 = 0
\end{equation}
\textbf{Conclusion:} $I(y_t; y_{<t} | X) \approx 0$ implies that $y_t$ is conditionally independent of $y_{<t}$ given $X$. Therefore, the complex autoregressive probability $P(y_t | y_{<t}, X)$ mathematically degenerates to the marginal $P(y_t | X)$. This theoretically justifies the use of the NAR architecture for reranking.
\end{proof}

\subsection{Proof of Stability: $\epsilon$-Bound Theorem}
\label{sec:proof_stability}

Here we provide the proof for the Stability Theorem stated in Section 3.2, establishing that knowledge distillation provides a safety guarantee for the NAR student against ranking errors.

\textbf{Definition A.2 (Probabilistic Ranking Margin).}
We define the ranking confidence of the Teacher model as the probability gap between a relevant item $y_i$ and a less relevant item $y_j$. A robust Teacher satisfies a margin condition $\delta > 0$:
\begin{equation}
    P_T(y_i|X) - P_T(y_j|X) \ge \delta, \quad \forall (i, j) \in \mathcal{P}_{ordered}
\end{equation}
where $\mathcal{P}_{ordered}$ is the set of pairs where item $i$ is strictly better than item $j$.

\textbf{Theorem A.3 ($\epsilon$-Bound Ranking Stability).}
\textit{Let $P_T$ be a teacher distribution with margin $\delta$, and $P_S$ be a student distribution trained such that the KL-divergence $D_{KL}(P_T || P_S) \le \epsilon$. A sufficient condition to prevent a Ranking Flip (i.e., erroneously predicting $P_S(y_j) > P_S(y_i)$) is:}
\begin{equation}
    \delta > \sqrt{2\epsilon}
\end{equation}
\textit{Conversely, the probability of a flip occurring is strictly bounded by the ratio of the error to the margin.}

\begin{proof}
\textbf{Step 1: Relating KL to Total Variation.}
By Pinsker’s Inequality, the KL divergence bounds the Total Variation Distance (TVD) between the teacher ($P_T$) and student ($P_S$) distributions:
\begin{equation}
    \delta_{TV}(P_T, P_S) = \sup_{A} |P_T(A) - P_S(A)| \le \sqrt{\frac{1}{2} D_{KL}(P_T || P_S)} \le \sqrt{\frac{\epsilon}{2}}
\end{equation}
This inequality implies that for any single item outcome $y$, the pointwise probability deviation is bounded by:
\begin{equation}
    |P_T(y) - P_S(y)| \le \sqrt{\frac{\epsilon}{2}}
\end{equation}

\textbf{Step 2: Worst-Case Analysis.}
A \textbf{Ranking Flip} occurs if the student predicts $P_S(y_j) > P_S(y_i)$ despite the teacher's ground truth $P_T(y_i) \ge P_T(y_j) + \delta$.
In the worst-case scenario, the student underestimates the probability of the superior item $y_i$ and overestimates the inferior item $y_j$ by the maximum possible error margin.
Let $e = \sqrt{\epsilon/2}$. The condition for a flip is:
\begin{equation}
    P_S(y_i) < P_S(y_j)
\end{equation}
Substituting the error bounds relative to the teacher:
\begin{equation}
    (P_T(y_i) - e) < (P_T(y_j) + e)
\end{equation}
Rearranging terms:
\begin{equation}
    P_T(y_i) - P_T(y_j) < 2e = \sqrt{2\epsilon}
\end{equation}
However, by Definition A.2, we know the teacher imposes a margin $P_T(y_i) - P_T(y_j) \ge \delta$. A contradiction (and thus a flip) is physically impossible if the margin exceeds the error bound:
\begin{equation}
    \delta \ge \sqrt{2\epsilon} \implies \text{No Flip}
\end{equation}
\textbf{Conclusion:} To strictly guarantee stability, we require $\epsilon < \delta^2 / 2$. This proves that minimizing the distillation loss $\epsilon$ directly expands the stability region, making the Student asymptotically consistent with the Teacher's ordering logic.
\end{proof}

\section{Implementation Details}
\label{sec:appendix_implementation}

To ensure reproducibility, we detail the model architecture, training configuration, and hyperparameters used in our experiments. All models are implemented in TensorFlow and trained on a high-performance computing cluster.

\subsection{Model Architecture}
We adopt a shared encoder architecture for both the Teacher and the Student to extract representations from user behavior sequences and candidate items.

\begin{itemize}
    \item \textbf{Shared Encoder:} The encoder is a Transformer-based architecture consisting of $N_{enc}=4$ layers. It employs Multi-Head Self-Attention with $H=3$ heads, a hidden dimension of $d_{model}=256$, and an inner feed-forward dimension of $d_{ffn}=256$. The input sequence length is set to $L=30$, corresponding to the re-ranking candidate list size.
    
    \item \textbf{Teacher Model:} The Teacher employs a Transformer Decoder architecture initialized with $N_{teacher}=4$ layers. It utilizes causal masking to model the generative probability of the optimal sequence. During the distillation phase, the teacher provides soft labels to guide the student.
    
    \item \textbf{Student Model (Generator):} The Student is a lightweight Transformer Decoder consisting of only $N_{student}=1$ layer. It incorporates both causal self-attention and cross-attention mechanisms to interact with the encoder outputs efficiently. This shallow architecture significantly reduces inference latency while maintaining representation capability through distillation.
\end{itemize}

\subsection{Training Configuration}
The training process integrates Supervised Fine-Tuning (SFT) loss, Distillation loss, and the proposed LDRO reinforcement learning objective.

\begin{itemize}
    \item \textbf{Optimization:} We use an Adagrad-based optimizer for sparse embeddings with an initial accumulator value of 3.0. For dense parameters, we utilize a specialized dense optimizer with a learning rate of $9.0 \times 10^{-6}$ and decay rates $\beta_1=0.99, \beta_2=0.9999$.
    
    \item \textbf{Loss Weights:} The total loss is a weighted sum of the exposure loss ($\mathcal{L}_{exp}$), the intra-position BRP loss, the distillation loss ($\mathcal{L}_{distill}$), and the RL loss ($\mathcal{L}_{GRPO}$).
    \begin{equation}
        \mathcal{L}_{total} = \mathcal{L}_{exp} + 0.5 \cdot \mathcal{L}_{GRPO} + \mathcal{L}_{intra\_brp} + 1.0 \cdot \mathcal{L}_{distill}
    \end{equation}
    The distillation temperature is set to $\tau=1.0$.
    
    \item \textbf{LDRO Hyperparameters:}
        \item \textbf{Mini Batch Size ($\mathcal{B}$):} We set the mini-batch size to $\mathcal{B}=1024$ queries per training step.
        \item \textbf{Group Size ($G$):} We sample $G=12$ sequences for each input query to estimate the baseline.
        \item \textbf{KL Coefficients:} The KL divergence penalty coefficient is set to $\beta_{KL}=0.02$.
        \item \textbf{Entropy Bonus:} To encourage exploration, we use an entropy coefficient of $\beta_{entropy}=0.05$.
        \item \textbf{Objective Standardization:} We apply decoupled normalization to the rewards (CTR, LvTR, etc.) before weighted summation. The weights for the objectives are set uniformly to $1.0$.
    \end{itemize}

\subsection{Inference Strategy}
To ensure diversity and enable efficient exploration during inference, we employ a \textbf{Gumbel-Max Sampling} strategy rather than deterministic greedy decoding:

\begin{itemize}
    \item \textbf{Gumbel-Max Sampling:} We simulate sampling from the categorical distribution by adding Gumbel noise to the unnormalized log-probabilities (logits). For a given time step $t$, the next item $y_t$ is selected via:
    \begin{equation}
        y_t = \arg\max_{i \in \mathcal{V}} \left( \frac{z_{t,i}}{\tau} + g_{t,i} \right)
    \end{equation}
    where $z_{t,i}$ is the logit for candidate $i$, $\tau$ is the temperature parameter (set to 0.8), and $g_{t,i} \sim \text{Gumbel}(0, 1)$ is independent noise drawn from the standard Gumbel distribution. This transforms the sampling process into a deterministic optimization problem given the noise.

    \item \textbf{Vectorized Optimization:} To minimize inference latency, we implement a fully vectorized version of the Gumbel-Max trick. Specifically, we pre-compute the Gumbel noise tensors for the entire sequence length and batch size \textit{outside} the auto-regressive decoding loop. This design avoids expensive random number generation calls within the loop and allows us to use the highly optimized \texttt{argmax} operator instead of the slower \texttt{multinomial} sampling operation, significantly accelerating the step-wise generation.
\end{itemize}

\begin{table}[h]
\centering
\caption{Hyperparameters used in our experiments.}
\label{tab:hyperparams}
\begin{tabular}{l|c}
\toprule
\textbf{Hyperparameter} & \textbf{Value} \\
\midrule
Sequence Length ($L$) & 30 \\
Embedding Dimension ($d_{model}$) & 256 \\
Feed-Forward Dimension ($d_{inner}$) & 256 \\
Attention Heads ($H$) & 3 \\
Head Dimension ($d_k, d_v$) & 64 \\
Dropout Rate & 0.1 \\
\midrule
Teacher Layers & 4 \\
Student Layers & 1 \\
\midrule
GRPO Group Size ($G$) & 12 \\
KL Coefficient ($\beta_{KL}$) & 0.02 \\
Entropy Coefficient ($\beta_{entropy}$) & 0.05 \\
Distillation Temperature ($\tau$) & 1.0 \\
\bottomrule
\end{tabular}
\end{table}

\section{Metrics Definition for Figure 1}
\label{app:metrics}

To empirically validate the \textit{Unimodal Concentration Hypothesis}, we introduced three quantitative metrics in Figure 1 to characterize the distribution differences between open-ended NLP tasks and industrial reranking tasks. The detailed definitions are as follows:

\subsection{Top-1 Confidence}
This metric measures the certainty of the model's prediction at each decoding step. It is defined as the probability mass assigned to the highest-ranked token (item) in the vocabulary (candidate set) $\mathcal{V}$:
\begin{equation}
    \text{Conf}_{\text{top1}} = \frac{1}{L} \sum_{t=1}^{L} \max_{v \in \mathcal{V}} P(v | y_{<t}, x)
\end{equation}
A high Top-1 Confidence indicates that the probability distribution is sharply peaked, suggesting a deterministic decision process, whereas a low value implies high uncertainty or high entropy.

\subsection{Branching Factor}
The Branching Factor quantifies the effective number of viable next-token choices at each step. It is derived from the exponentiation of the entropy (Perplexity-like metric):
\begin{equation}
    \mathcal{B} = \frac{1}{L} \sum_{t=1}^{L} 2^{H(P_t)} = \frac{1}{L} \sum_{t=1}^{L} 2^{-\sum_{v \in \mathcal{V}} P(v)\log_2 P(v)}
\end{equation}
where $H(P_t)$ is the Shannon entropy of the prediction distribution at step $t$.
\begin{itemize}
    \item \textbf{High Branching Factor ($\approx 6.0$ for NLP):} Indicates a ``one-to-many'' mapping where multiple words are equally valid (e.g., creative writing).
    \item \textbf{Low Branching Factor ($\approx 2.0$ for Reranking):} Indicates a ``quasi-deterministic'' mapping where the optimal item is highly constrained by the context.
\end{itemize}

\subsection{Sequence Consistency}
This metric evaluates the stability of the generated sequence against the ground truth. Since exact matching is too strict for long sequences, we employ the \textbf{Jaccard Similarity} coefficient to measure the overlap between the set of tokens in the greedily generated sequence $\hat{Y}_{\text{greedy}}$ and the ground truth sequence $Y_{\text{truth}}$:
\begin{equation}
    \text{Seq. Consistency} = \frac{|\hat{Y}_{\text{greedy}} \cap Y_{\text{truth}}|}{|\hat{Y}_{\text{greedy}} \cup Y_{\text{truth}}|}
\end{equation}
As shown in Figure 1(b), the Reranking task exhibits high consistency ($0.75$), confirming that given a specific user context, the optimal ranking order is largely unique and reproducible. In contrast, NLP tasks show low consistency ($0.23$) due to the intrinsic linguistic diversity (i.e., multiple valid ways to express the same meaning).

\end{document}